\documentclass[reqno, eucal]{amsart}

\usepackage[mathscr]{eucal} 
\usepackage[dvips]{graphicx}
\usepackage[dvips]{color}
\usepackage{amsmath,amsfonts,amssymb,amsthm,amscd}
\usepackage{epic}
\usepackage{eepic}
\usepackage{longtable}
\usepackage{array}
\usepackage{tabmac}

\newtheorem{theorem}{Theorem}

\newtheorem{proposition}[theorem]{Proposition}
\newtheorem{corollary}[theorem]{Corollary}

\theoremstyle{definition}

\newtheorem{example}[theorem]{Example}

\newcommand{\Z}{{\mathbb Z}}

\newcommand{\C}{{\mathbb C}}

\newcommand{\Aff}{\mathrm{Aff}}
\newcommand{\Cx}{\C^\times}
\newcommand{\la}{\lambda}
\newcommand{\ot}{\otimes}
\newcommand{\tab}[1]{\vcenter{\tableau[sby]{#1}}}
\newcommand{\te}{\tilde{e}}
\newcommand{\tf}{\tilde{f}}

\begin{document}

\title[Set-theoretical reflection equation]{Set-theoretical solutions to the \\
reflection equation associated to the \\
quantum affine algebra of type $\boldsymbol{A^{(1)}_{n-1}}$}

\author{Atsuo Kuniba}
\address{Atsuo Kuniba, Institute of Physics, 
University of Tokyo, Komaba, Tokyo 153-8902, Japan}
\email{atsuo.s.kuniba@gmail.com}

\author{Masato Okado}
\address{Masato Okado, Department of Mathematics, Osaka City University, 
Osaka, 558-8585, Japan}
\email{okado@sci.osaka-cu.ac.jp}

\maketitle

\vspace{0.5cm}
\begin{center}{\bf Abstract}
\end{center}
A trick to obtain a systematic solution 
to the set-theoretical reflection equation is presented from a known one to
the Yang-Baxter equation. Examples are given from crystals and geometric crystals
associated to the quantum affine algebra of type $A^{(1)}_{n-1}$.

\vspace{0.4cm}

\section{Introduction}

Principal features of integrable systems in the bulk and at the boundary
are the Yang-Baxter equation \cite{Bax} 
and the reflection equation \cite{Ch,Kul,Sk}, respectively.
In their original formulation in quantum field theory or statistical mechanics, 
they are cubic and quartic relations of the form $RRR = RRR$ and $RKRK=KRKR$,
where $R$ and $K$ are {\em matrices} encoding the interactions 
in the bulk and at the boundary of the system.
By now extensive results and knowledge on these equations and solutions 
have been accumulated  
especially via their connection to the theory of quantum groups
\cite{D,Ji}.

The Yang-Baxter and reflection equations can also be formulated in a wider context 
where the linear operators $R$ and $K$ are replaced 
by transformations of various kind
such as bijection between sets, birational maps between varieties and so forth.
Such extensions, originally suggested for the Yang-Baxter equation \cite{D2}, 
are often called the {\em set-theoretical} ones.
See for instance \cite{V} for a guide to the set-theoretical $R$'s, and 
\cite{KOY,CCZ,CZ,SVW} 
for some concrete examples of set-theoretical $K$'s.
By definition a set-theoretical solution to the 
reflection equation means a pair of 
set-theoretical $R$ and $K$.

In this paper we present new  
set-theoretical solutions to the reflection equation.
They are most natural and systematic examples 
associated with the Drinfeld-Jimbo quantum affine algebra $U_q(A^{(1)}_{n-1})$. 
There are two versions $(R_{B,B'},K_B)$ and $(\mathcal{R},\mathcal{K})$
which originate in the crystals \cite{Kas} and 
the geometric crystals \cite{BK} related to the Kirillov-Reshetikhin modules
of $U_q(A^{(1)}_{n-1})$, respectively.
For Kirillov-Reshetikhin modules of quantum affine algebras and their
rich background, see e.g. \cite{HKOTT}. 

In the first solution, 
$R_{B,B'}: B \times B' \rightarrow B' \times B$ 
and $K_B: B \rightarrow B^\vee$ are bijections among
the indicated finite sets called {\em crystals}.
Here $^\vee$ denotes the dual whose detail is given in Section 2.
On the other hand, $\mathcal{R}$ and $\mathcal{K}$  
treated in Section 3 are 
birational maps between varieties.
The solution $(R_{B,B'},K_B)$  
can be recovered from $(\mathcal{R},\mathcal{K})$ 
by the procedure called {\em tropicalization} (cf. \cite{SS}) 
or {\em ultra-discretization} 
in another terminology frequently used for  
integrable cellular automata (cf. \cite{IKT, KOY}).

The set-theoretical $R_{B,B'}$ is a
{\em combinatorial $R$-matrix}\footnote{
Although $R_{B,B'},K_B, \mathcal{R},\mathcal{K}$ are
all set-theoretical, we refer to them with ``matrix" 
following the terminology in \cite{KMN1, Fr2}.} in the sense of \cite[Def.4.2.1]{KMN1}
and corresponds to a quantum $R$-matrix at $q=0$.
It preserves the crystal structure (sort of colored oriented graph) 
on $B \times B'$ and $B' \times B$ inheriting
the quantum group symmetry of the relevant quantum $R$-matrices.
It is a highly nontrivial bijection described by an elegant tableau combinatorics \cite{Sh}.
The set-theoretical $\mathcal{R}$ is called 
the {\em geometric $R$-matrix} \cite{Fr2}, which covers a few
special cases known earlier \cite[sec. A.3]{KOY}.
A contribution of this paper is the construction of  
$K_B$ and $\mathcal{K}$ from $R_{B,B'}$ and $\mathcal{R}$, respectively,
in a systematic manner.
As already mentioned, $K_B$ is deducible in principle from $\mathcal{K}$ 
by ultra-discretization.
However this procedure just 
provides the former with a formidably complicated 
piecewise linear formula which can be handled effectively only on computers.
Thus it is worth while to describe $K_B$ independently by a neat
combinatorial algorithm.
This will be done in Section 2 by attributing $K_B$ 
to the combinatorial $R$-matrix $R_{B^\vee, B}$.

Let us digest our approach to the set-theoretical reflection equation 
along the example $(R_{B,B'},K_B)$.
Recall that the Yang-Baxter equation is depicted by the 
world lines of three particles on a plane.
Now we consider the Yang-Baxter equality 
involving two pairs of particles,  hence {\em four} world lines in total,
in the vicinity of a {\em virtual boundary}. 
Arrange two incoming particles within a pair so that their motion 
become mirror image of each other with respect to the boundary.
See Figure \ref{fig1}, where the boundary is depicted by broken lines.
The mirror images are realized by switching to the {\em dual} crystals/tableaux 
signified by the superscript ${}^\vee$.
Then it can be shown that the reflection symmetry 
of the incoming state via ${}^\vee$ 
persists throughout the whole scattering event.
This claim is formulated as 
Proposition \ref{prop:dual} and Corollary \ref{cor:dual}.
The resulting Yang-Baxter equality restricted on either side of the mirror 
is nothing but a set-theoretical reflection equation.
This trick of using the virtual boundary works efficiently   
in a set-theoretical situation like here and \cite{CCZ,CZ}.
However in general, it does not extend naively to the original (quantum) 
version of the reflection equation where $R$ and $K$ are linear operators.

We remark that the three dimensional analogue of the reflection equation
known as the {\em tetrahedron reflection equation} 
$R_{456} R_{489}
 K_{3579} R_{269} R_{258}
 K_{1678} K_{1234}
= K_{1234}
 K_{1678}
 R_{258} R_{269}
 K_{3579} R_{489} R_{456}$ \cite{IK}
also admits a similar triad of quantum, combinatorial and birational versions of 
solutions in which the latter two are set-theoretical \cite[Tab.1]{KO1}.


\section{Kirillov-Reshetikhin crystal for $A^{(1)}_{n-1}$}

We recall the Kirillov-Reshetikhin crystal, KR crystal for short, $B^{k,l}$ for type
$A^{(1)}_{n-1}$, where $1\le k\le n-1,l\ge1$. It is a crystal basis, in the sense of
Kashiwara \cite{Kas}, of a certain finite-dimensional module, called Kirillov-Reshetikhin
module over the quantized enveloping algebra of affine type $A^{(1)}_{n-1}$. 
See e.g. \cite{HKOTT}.
 As a set $B^{k,l}$ consists of semi-standard
tableaux of $k\times l$ rectangular shape with letters from $\{1,2,\ldots,n\}$. 
On $B^{k,l}$ applications of Kashiwara operators $\te_i,\tf_i$ for $0\le i\le n-1$ are
defined \cite{KMN2,Sh}. A few examples follow.

\begin{example} \label{ex:action}
\begin{enumerate}
\item $n=4,B^{1,5}$
\[
\te_2\,\small{\tab{1&2&3&3&4}}=\small{\tab{1&2&2&3&4}},\quad
\tf_0\,\small{\tab{1&2&3&3&4}}=\small{\tab{1&1&2&3&3}}.
\]
\item $n=6,B^{4,1}$
\[
\te_5\,\small{\tab{1\\3\\4\\6}}=\small{\tab{1\\3\\4\\5}},\quad
\tf_3\,\small{\tab{1\\3\\4\\6}}=0\;(\text{some element outside $B^{4,1}$}).
\]
\item $n=6,B^{4,3}$
\[
\te_3\,\small{\tab{1&1&3\\2&2&4\\3&4&5\\5&5&6}}
=\small{\tab{1&1&3\\2&2&4\\3&3&5\\5&5&6}},\quad
\tf_0\,\small{\tab{1&1&3\\2&2&4\\3&4&5\\5&5&6}}
=\small{\tab{1&1&1\\2&2&3\\3&4&4\\5&5&5}}.
\]
\end{enumerate}
\end{example}
For details see e.g. \cite{O:Memoirs}. For an element $b$ of a KR crystal, we set
\[
\varepsilon_i(b)=\max\{k\in\Z_{\ge0}\mid \te_ib\ne0\},\quad
\varphi_i(b)=\max\{k\in\Z_{\ge0}\mid \tf_ib\ne0\}.
\]

Let $B_1,B_2$ be KR crystals. The Cartesian product of $B_1$ and $B_2$ denoted by
$B_1\ot B_2=\{b_1\ot b_2\mid b_1\in B_1,b_2\in B_2\}$ is also endowed with 
the crystal structure by
\begin{align}
\te_i(b_1 \otimes b_2) & := \begin{cases}
\te_i b_1 \otimes b_2 & \text{if } \varepsilon_i(b_1) > \varphi_i(b_2), \\
b_1 \otimes \te_i b_2 & \text{if } \varepsilon_i(b_1) \leqslant \varphi_i(b_2)\,,
\end{cases} \label{eq:e for two factors}
\\ \tf_i(b_1 \otimes b_2) & := \begin{cases}
\tf_i b_1 \otimes b_2 & \text{if } \varepsilon_i(b_1) \geqslant \varphi_i(b_2), \\
b_1 \otimes \tf_i b_2 & \text{if } \varepsilon_i(b_1) < \varphi_i(b_2)\,.
\end{cases} \label{eq:f for two factors} 
\end{align}
We use this convention of the tensor product, which is opposite from
the one in \cite{Kas}.

For KR crystals $B_1,B_2$ there exists a unique bijection 
$R_{B_1,B_2}:B_1\ot B_2\rightarrow B_2\ot B_1$, called the combinatorial $R$-matrix,
that commutes with $\te_i,\tf_i$ for any $i$ \cite{KMN1}.
Uniqueness follows from the fact that $B_1\ot B_2$ is connected, namely, any element
of $B_1\ot B_2$ can be reached from a fixed element by applying $\te_i$'s or $\tf_i$'s. 
Explicitly, the combinatorial $R$-matrix for type $A^{(1)}_{n-1}$ can be calculated
by the so-called tableau product \cite{Fu}. Since the tensor product of two
representations corresponding to rectangular shapes is multiplicity free and the 
operators $\te_i,\tf_i$ ($i\ne0$) commute with the operations to construct the tableau 
product, for 
the tableau product $b_1\cdot b_2$ ($b_1\in B^{k_1,l_1},b_2\in B^{k_2,l_2}$)  there is a
unique pair $(\tilde{b}_2,\tilde{b}_1)\in B^{k_2,l_2}\times B^{k_1,l_1}$ such that
$b_1\cdot b_2=\tilde{b}_2\cdot\tilde{b}_1$. In this case $R(b_1\ot b_2)=\tilde{b}_2
\ot\tilde{b}_1$ \cite{Sh}.
Combinatorial $R$-matrices satisfy the Yang-Baxter equation as a map from
$B_1\ot B_2\ot B_3$ to $B_3\ot B_2\ot B_1$:
\[
(R_{B_2,B_3}\ot1)(1\ot R_{B_1,B_3})(R_{B_1,B_2}\ot 1)
=(1\ot R_{B_1,B_2})(R_{B_1,B_3}\ot1)(1\ot R_{B_2,B_3}).
\]
\begin{example} \label{ex:R}
Let $n=6$, and 
$b_1=\small{\tab{1&3&4\\2&6&6}}\in B^{2,3},
b_2=\small{\tab{1&1&3\\2&2&4\\3&4&5\\5&5&6}}\in B^{4,3}$. 
Then the tableau product is given by
\[
b_1\cdot b_2=\small{\tab{1&1&1&3&4&5\\2&2&2&4\\3&3&5\\4&5&6\\6&6}}.
\]
The only pair of tableaux $(\tilde{b}_2,\tilde{b}_1)\in B^{4,3}\times B^{2,3}$ such that
$b_1\cdot b_2=\tilde{b}_2\cdot\tilde{b}_1$ is given by
$\tilde{b}_2=\small{\tab{1&2&2\\2&3&3\\4&5&5\\6&6&6},
\tilde{b}_1=\tab{1&1&3\\4&4&5}}$. Hence we have 
$R(b_1\ot b_2)=\tilde{b}_2\ot\tilde{b}_1$. See \cite[Example 3.7]{O:Memoirs} and
the explanations above for more details. 
(Note however that the convention for the tensor
product there is opposite.)
\end{example}

The notion of dual crystal is given in \cite[Section 7.4]{Kas2}. Let $B$ be a crystal. 
Then there is a crystal denoted $B^\vee$ obtained from $B$ defined 
by 
$B^\vee=\{b^\vee\mid b\in B\}$ with
\begin{equation} \label{eq:dual crystal}
\te_ib^\vee= (\tf_ib)^\vee, \qquad
\tf_ib^\vee = (\te_ib)^\vee.
\end{equation}
Then there is an isomorphism
$(B_1\otimes B_2)^\vee \cong B_2^\vee \otimes B_1^\vee$ given by
$(b_1\otimes b_2)^\vee \mapsto b_2^\vee \otimes b_1^\vee$.
For a KR crystal $B^{k,l}$ of type $A^{(1)}_{n-1}$ we have $(B^{k,l})^\vee=B^{n-k,l}$
and $b^\vee$ is obtained from $b$ by replacing each column of a rectangular tableau
with its compliment in $\{1,\ldots,n\}$ and reversing the order of the columns.

\begin{example} \label{ex:dual}
\begin{enumerate}
\item When $B=B^{1,l}$ one can use a sequence of nonnegative integers 
$x(b)=(x_1,\ldots,x_n)$ to parametrize a crystal element $b$, where $x_i$ stands for
the number of $i$ in the one-row tableau $b$. Then $b^\vee\in B^{n-1,l}$
is given by the $(n-1)\times l$ tableau such that the number of columns
missing $i$ is $x_i$.
\item Let $n=6$.
For $b\in\small{\tab{1&1&3\\2&2&4\\3&4&5\\5&5&6}}\in B^{4,3}$ we have
$b^\vee=\small{\tab{1&3&4\\2&6&6}}\in B^{2,3}$. $\tf_3 b^\vee$ is given by
$(\te_3 b)^\vee$ where $\te_3 b$ is in Example \ref{ex:action}(3).
\end{enumerate}
\end{example}

\begin{proposition} \label{prop:dual}
If $R_{B_1,B_2}(b_1\ot b_2)=c_2\ot c_1$, then 
$R_{B_2^\vee,B_1^\vee}(b_2^\vee\ot b_1^\vee)=c_1^\vee\ot c_2^\vee$.
\end{proposition}

\begin{proof}
For any isomorphism of crystals $\psi:B\rightarrow B'$, the map $\psi^\vee$
defined by $\psi^\vee=\vee\circ\psi\circ\vee$ becomes an isomorphism from
$B^\vee$ to $(B')^\vee$. For instance,
\[
\psi^\vee(\te_ib^\vee)=\psi^\vee((\tf_ib)^\vee)=(\psi(\tf_ib))^\vee=(\tf_i\psi(b))^\vee
=\te_i(\psi(b))^\vee=\te_i\psi^\vee(b^\vee).
\]
Combining the isomorphism $\iota_{i,j}:(B_i\ot B_j)^\vee\rightarrow B_j^\vee\ot
B_i^\vee$, one finds
\[
(\iota_{2,1}\circ\vee\circ R_{B_1,B_2}\circ\vee\circ\iota_{1,2}^{-1})
(b_2^\vee\ot b_1^\vee)=c_1^\vee\ot c_2^\vee
\]
also gives an isomorphism. Since $B_2^\vee\ot B_1^\vee$ is connected, this map agrees with
$R_{B_2^\vee,B_1^\vee}$.
\end{proof}

\begin{corollary} \label{cor:dual}
We have:
\begin{enumerate}
\item If $R_{B^\vee,B}(b_1^\vee\ot b_2)=c_2\ot c_1^\vee$, then 
$R_{B^\vee,B}(b_2^\vee\ot b_1)=c_1\ot c_2^\vee$.
\item $R_{B^\vee,B}(b^\vee\ot b)=c\ot c^\vee$.
\end{enumerate}
\end{corollary}

\begin{proof}
For (1) set $B_1=B^\vee,B_2=B,b_1=b_1^\vee,c_1=c_1^\vee$ in Proposition \ref{prop:dual}. 
For (2) set $b_1$ and $b_2$ to be $b$ in (1). Since $c_1=c_2$, we have the desired result.
\end{proof}
Using $b$ and $c^\vee$ in Corollary \ref{cor:dual}(2), we define the combinatorial
$K$-matrix $K_B$ by
\begin{align}\label{bango}
K_B: B \rightarrow B^\vee,\quad b \mapsto c^\vee.
\end{align}

\begin{example} \label{ex:K}
\begin{enumerate}
\item We continue considering an example in Example \ref{ex:dual}(1). Suppose 
$R_{(B^{1,l})^\vee,B^{1,m}}$ sends $b^\vee\ot c$ to $\tilde{c}\ot\tilde{b}^\vee$.
Set $x(b)=(x_i)_{i=1}^n,x(c)=(y_i)_{i=1}^n$. Then we have 
\[
x(\tilde{b})_i=x_i+p_{i+1}-p_i,\quad y(\tilde{c})_i=y_i+p_{i+1}-p_i,
\]
where $p_i=\min(x_i,y_i)$ and the index $i$ should be considered modulo $n$.
See \cite[(2.2)]{KOY}. Note that our convention of the tensor product is opposite
from there. Set $l=m,x_i=y_i$ for all $i$. Then we find $x(\tilde{b})_i=x(\tilde{c})_i=x_{i+1}$.
Hence, $K_{B^{1,l}}(b)=\tilde{b}$ where $x(\tilde{b})_i=x_{i+1}$.
\item We take an example from Example \ref{ex:R}. Namely, we set 
$n=6,B=B^{4,3}$, then $B^\vee=B^{2,3}$. The image of
$\small{\tab{1&3&4\\2&6&6}\ot\tab{1&1&3\\2&2&4\\3&4&5\\5&5&6}}$ by $R$ is 
$\small{\tab{1&2&2\\2&3&3\\4&5&5\\6&6&6}\ot\tab{1&1&3\\4&4&5}}$. Corollary \ref{cor:dual}
is confirmed in this example. Hence, the image of 
$\small{\tab{1&1&3\\2&2&4\\3&4&5\\5&5&6}}$ by $K_B$ is defined to be 
$\small{\tab{1&1&3\\4&4&5}}$.
\end{enumerate}
\end{example}

\begin{theorem} \label{th:main1}
The combinatorial $R$-matrix and $K$-matrix satisfy 
the set-theoretical reflection equation.
\[
R_{B_2^\vee,B_1^\vee}(K_{B_2}\ot1)R_{B_1^\vee,B_2}(K_{B_1}\ot1)
=(K_{B_1}\ot1)R_{B_2^\vee,B_1}(K_{B_2}\ot1)R_{B_1,B_2}
\]
\end{theorem}

\begin{proof}
Figure \ref{fig1} illustrates the proof. Each side represents the application of each
side of 
\begin{equation} \label{proof}
R_{B_2^\vee,B_1^\vee}^{3,4}R_{B_1,B_2}^{1,2}R_{B_2^\vee,B_2}^{2,3}R_{B_1^\vee,B_2}^{3,4}
R_{B_2^\vee,B_1}^{1,2}R_{B_1^\vee,B_1}^{2,3}
=R_{B_1^\vee,B_1^\vee}^{2,3}R_{B_1^\vee,B_2}^{1,2}R_{B_2^\vee,B_1}^{3,4}
R_{B_2^\vee,B_2}^{2,3}R_{B_2^\vee,B_1^\vee}^{1,2}R_{B_1,B_2}^{3,4}
\end{equation}
on $c^\vee\ot b^\vee\ot b\ot c\in B_2^\vee\ot B_1^\vee\ot B_1\ot B_2$. Here
superscripts in the $R$-matrices stands for the positions of the components on which
it acts. Equation \eqref{proof} itself is verified by successive use of the Yang-Baxter
equation, where the transitions of elements upon applications of $R$-matrices are
depicted in Figure \ref{fig1} using Proposition \ref{prop:dual} and Corollary \ref{cor:dual}.
Eventually, we obtain $c_3=c_6,b_3=b_6,b_3^\vee=b_6^\vee,c_3^\vee=c_6^\vee$.
Just viewing the right parts of both sides of Figure \ref{fig1} and recalling the definition
of the $K$-matrix, we finish the proof of the reflection equation.
\end{proof}

\begin{figure}[h] 
\[
\begin{picture}(300,185)(-32,-5)
\put(0,20){
\multiput(0,-8)(0,4){41}{\line(0,1){2}}
\put(0,0){\put(-30,0){\vector(2,3){95}}\put(30,0){\vector(-2,3){95}}}
\put(0,32){\put(-69,18){\vector(2,1){143}}\put(69,18){\vector(-2,1){143}}}
\put(-76,146){$b_3$}\put(67,146){$b_3^\vee$}
\put(-88,121){$c_3$}\put(78,121){$c_3^\vee$}
\put(-25,102){$c_2$}\put(15,102){$c_2^\vee$}
\put(-39,79){$b_2$}\put(30,79){$b_2^\vee$}
\put(-16,71){$c_1^\vee$}\put(6,71){$c_1$}
\put(-22,52){$b_1$}\put(12,52){$b_1^\vee$}
\put(-82,45){$c^\vee$} \put(-99,40){$B_2^\vee$}
   \put(73,45){$c$} \put(81,40){$B_2$}
\put(-40,-10){$b^\vee$}\put(-57,-17){$B_1^\vee$}
   \put(32,-10){$b$}\put(39,-17){$B_1$}
}
\put(111,85){$=$}
\put(235,20){
\multiput(0,-8)(0,4){41}{\line(0,1){2}}
\put(-60,-1){\vector(2,3){95}}\put(60,-1){\vector(-2,3){95}}
\put(-73,12){\vector(2,1){142}}\put(73,12){\vector(-2,1){142}}
\put(-44,145){$b_6$}  \put(37,145){$b_6^\vee$}
\put(-21,76){$b_5^\vee$} \put(12,76){$b_5$}
\put(-80,86){$c_6$} \put(72,86){$c_6^\vee$}
\put(-13,58.5){$c_5$}\put(4,59){$c_5^\vee$}
\put(-41,46){$b_4^\vee$}\put(31,46){$b_4$}
\put(-18,31){$c_4^\vee$}\put(13,31){$c_4$}
\put(-85,7){$c^\vee$}\put(-104,2){$B_2^\vee$} \put(77,7){$c$}\put(86,2){$B_2$}
\put(-68,-11){$b^\vee$}\put(-86,-17){$B_1^\vee$}
\put(63,-10){$b$}\put(72,-17){$B_1$}
}
\end{picture}
\]
\caption{Proof of the reflection equation.} \label{fig1}
\end{figure}
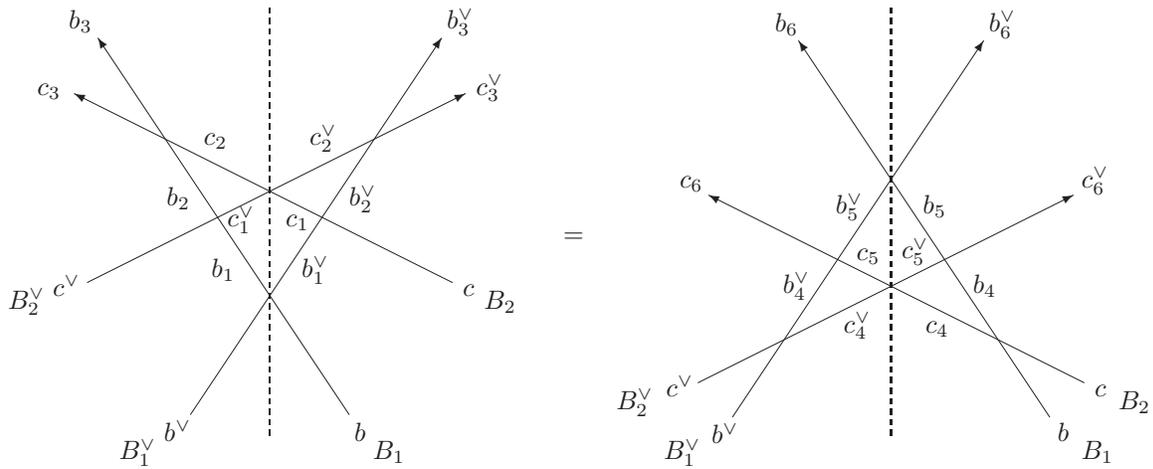

A similar argument has also been given in \cite{CZ}.

\begin{example} \label{ex:reflection}
$n=5$, $B_1 = B^{1,2}, B_2=B^{2,1}$, which implies 
$B_1^\vee = B^{4,2}, B_2^\vee =B^{3,1}$.
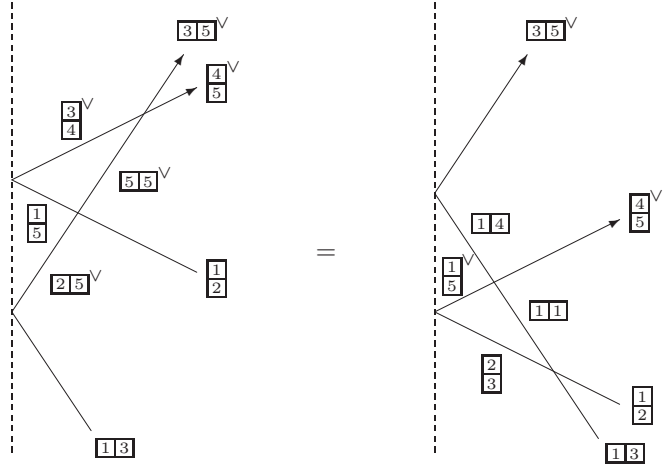
\[
\begin{picture}(300,185)(-32,5)
\put(0,20){
\multiput(0,-8)(0,4){43}{\line(0,1){2}}
\put(0,15){\put(0,30){\vector(2,3){65}}\put(0,30){\line(2,-3){30}}}
\put(0,5){\put(0,90){\vector(2,1){70}}\put(0,90){\line(2,-1){70}}}
\put(63,149){$\tiny{\tab{3&5}^\vee}$}
\put(74,129){$\tiny{\tab{4\\5}^\vee}$}
\put(19,115){$\tiny{\tab{3\\4}^\vee}$}
\put(41,92){$\tiny{\tab{5&5}^\vee}$}
\put(6,76){$\tiny{\tab{1\\5}}$}
\put(15,53){$\tiny{\tab{2&5}^\vee}$}
\put(74,55){$\tiny{\tab{1\\2}}$}\put(32,-9){$\tiny{\tab{1&3}}$}
}
\put(115,85){$=$}
\put(160,20){
\multiput(0,-8)(0,4){43}{\line(0,1){2}}
\put(0,15){\put(0,30){\vector(2,1){70}}\put(0,30){\line(2,-1){70}}}
\put(0,0){\put(0,90){\vector(2,3){35}}\put(0,90){\line(2,-3){62}}}
\put(35,149){$\tiny{\tab{3&5}^\vee}$}
\put(3,56){$\tiny{\tab{1\\5}^\vee}$}
\put(74,80){$\tiny{\tab{4\\5}^\vee}$}
\put(14,76){$\tiny{\tab{1&4}}$}
\put(36,43){$\tiny{\tab{1&1}}$}
\put(18,19){$\tiny{\tab{2\\3}}$}
\put(75,7){$\tiny{\tab{1\\2}}$}
\put(65,-11){$\tiny{\tab{1&3}}$}
}
\end{picture}
\]
\end{example}

For a KR crystal $B$ we can define its affinization $\Aff(B)=\{z^db\mid b\in B,d\in\Z\}$.
On $\Aff(B)$ Kashiwara operators act as $\te_i(z^db)=z^{d+\delta_{i0}}(\te_ib),
\tf_i(z^db)=z^{d-\delta_{i0}}(\tf_ib)$. Combinatorial $R_{B_1,B_2}$-matrix is also upgraded by 
introducing the energy function $H$ \cite{KMN1} as
\begin{eqnarray*}
R_{B_1,B_2}:\Aff(B_1)\ot\Aff(B_2)&\longrightarrow&\Aff(B_2)\ot\Aff(B_1)\\
z^{d_1}b_1\ot z^{d_2}b_2&\longmapsto&z^{d_2+H(b_1\ot b_2)}c_2\ot z^{d_1-H(b_1\ot b_2)}c_1.
\end{eqnarray*}
This version of the combinatorial $R$-matrices also satisfy the Yang-Baxter equation. 
Hence, by defining $K_B(z^db)=z^{-d-H(b^\vee\ot b)}c^\vee$ 
when $R_{B^\vee,B}(z^{-d}b^\vee\ot z^d b)=z^{d+H(b^\vee\ot b)}c\ot z^{-d-H(b^\vee\ot b)}
c^\vee$, upgraded combinatorial $R$-matrices and $K$-matrices satisfy the reflection
equation.
This generalizes \cite[sec.2.3]{KOY}.

\section{Geometric crystal for $A_{n-1}^{(1)}$}

Geometric crystal is a notion introduced by Berenstein and Kazhdan in \cite{BK} as an 
algebro-geometric analogue of crystal. 
A geometric crystal of type $A_{n-1}^{(1)}$ is a pentad 
$(X,\gamma,\varphi_i,\varepsilon_i,e_i)$ where $X$ is an irreducible complex algebraic
variety, $\gamma:X\rightarrow (\Cx)^n$ is a rational map, for $i\in\Z/n\Z$
$\varphi_i,\varepsilon_i:X\rightarrow\Cx$ are rational functions and $e_i:\Cx\times X
\rightarrow X$ is a rational action. These data must satisfy further relations.
For instance, denoting the image $e_i(c,x)$ by $e_i^c(x)$, the actions $e_i^{c_1},e_j^{c_2}$
should satisfy $e_i^{c_1}e_j^{c_2}=e_j^{c_2}e_i^{c_1}$ if $|i-j|>1$, 
$e_i^{c_1}e_j^{c_1c_2}e_i^{c_2}=e_j^{c_2}e_i^{c_1c_2}e_j^{c_1}$ if $|i-j|=1$.

For $1\le k\le n-1$, Frieden \cite{Fr1} gave a geometric crystal structure 
of type $A_{n-1}^{(1)}$ on 
$\mathrm{Gr}(n-k,n)\times\Cx$ where $\mathrm{Gr}(m,n)$
is the Grassmannian of $m$-dimensional subspaces in $\C^n$.
Introduce the space of ``rational $k$-rectangle'' by
$\mathbf{T}_k=(\Cx)^{R_k}\times\Cx$
where 
\[
R_k=\{(i,j)\mid 1\le i\le k,i\le j\le i+n-k-1\}.
\]
An element of $\mathbf{T}_k$ is denoted by $(X,s)$. There is an open embedding
of $\mathbf{T}_k$ into $\mathrm{Gr}(n-k,n)\times\Cx$. In the construction
in \cite{Fr2} it is also important that there exists an injection $g$ from $\mathbf{T}_k$ to
$GL_n(\C(\la))$ satisfying
\[
g(e_i^c(X,s))=\hat{x}_i\!\left(\frac{c-1}{\varphi_i(X,s)}\right) g(X,s)\;
\hat{x}_i\!\left(\frac{c^{-1}-1}{\varepsilon_i(X,s)}\right),
\]
where $\hat{x}_i(a)=I+a\la^{-1}E_{i,i+1}$ for $i\in\Z/n\Z$, $I$ is the identity matrix
and $E_{ij}$ stands for the $(i,j)$ matrix unit.

Like the tensor product in crystals, there is a product of 
$(X,s)$ and $(Y,t)$ denoted by $(X,s)\times(Y,t)$ such that $g((X,s)\times(Y,t))=
g(X,s)g(Y,t)$. Thanks to this map $g$, the product $(X,s)\times(Y,t)$ acquires the 
structure of the geometric crystal.

\begin{theorem}[\cite{Fr2}]
There exists a birational map 
$\mathcal{R}:\mathbf{T}_{k_1}\times\mathbf{T}_{k_2}\rightarrow
\mathbf{T}_{k_2}\times\mathbf{T}_{k_1}$ that commutes with $e_i^c$. Moreover, on 
$\mathbf{T}_{k_1}\times\mathbf{T}_{k_2}\times\mathbf{T}_{k_3}$ 
the Yang-Baxter equation
\[
\mathcal{R}_{12}\mathcal{R}_{23}\mathcal{R}_{12}
=\mathcal{R}_{23}\mathcal{R}_{12}\mathcal{R}_{23}
\]
is satisfied. Here $\mathcal{R}_{ij}$ means that $\mathcal{R}$ acts
on the $i$-th and $j$-th components and the other one trivially.
\end{theorem}

\begin{theorem}[\cite{Fr2}]
There exists a duality map $D:\mathbf{T}_k\rightarrow \mathbf{T}_{n-k}$ satisfying
$e_i^c\circ D=D\circ e_i^{c^{-1}}$ and 
\[
D((X,s)\times(Y,t))=D(Y,t)\times D(X,s).
\]
\end{theorem}

In terms of the Grassmannian, the duality map $D$ corresponds to the orthogonal
compliment with respect to a certain bilinear form of the ambient vector space.
In the image of $g$, $g(D(X,s))$ is closely related to the matrix inverse of $g(X,s)$.

From these two theorems, one obtains similar claims to Proposition \ref{prop:dual}
and Corollary \ref{cor:dual}. Hence, by defining $\mathcal{K}:\mathbf{T}_k\rightarrow
\mathbf{T}_{n-k},(X,s)\mapsto(\tilde{X}^\vee,s)$ when 
$\mathcal{R}:\mathbf{T}_{n-k}\times\mathbf{T}_k\rightarrow
\mathbf{T}_k\times\mathbf{T}_{n-k},(X^\vee,s)\times(X,s)\mapsto(\tilde{X},s)\times(\tilde{X}^\vee,s)$
where $D(X,s)=(X^\vee,s)$, one can show

\begin{theorem} \label{th:main2}
The geometric $R$-matrix and $K$-matrix satisfy the 
set-theoretical reflection equation.
\[
\mathcal{R}\mathcal{K}_1\mathcal{R}\mathcal{K}_1
=\mathcal{K}_1\mathcal{R}\mathcal{K}_1\mathcal{R}
\]
\end{theorem}

In literature, there is a procedure called ultra-discretization among mathematical 
physicists
or tropicalization among mathematicians, turning positive rational maps 
into piecewise linear maps. 
See. e.g.  \cite[sec.4.1]{IKT}.
It is known in \cite{Fr2} that $\mathcal{R}$ and $D$ are
positive maps. Upon use of ultra-discretization Theorem \ref{th:main1} is reconfirmed.

\begin{example} 
In this example we use coordinate variables $(x_{ij})$ ($(i,j)\in R_k$) to represent a point in
$\mathbf{T}_k$. Upon ultra-discretization, $x_{ij}$ counts the number of $j$
in the $i$-th row and $t$ the length of a tableau.
\begin{enumerate}
\item Let $k_1=n-1,k_2=1$. The geometric $R$-matrix $\mathcal{R}:
\mathbf{T}_{n-1}\times\mathbf{T}_1\rightarrow\mathbf{T}_1\times\mathbf{T}_{n-1},(X^\vee,s)\times(Y,t)\mapsto
(\tilde{Y},t)\times(\tilde{X}^\vee,s)$ is given by
\[
\tilde{x}_i=x_i\frac{P_{i+1}}{P_i},\quad \tilde{y}_i=y_i\frac{P_{i+1}}{P_i}\quad\text{where }
P_i=x_i+y_i.
\]
Here $x_i=x_{1i}$ are for $X\in\mathbf{T}_1$, etc. By ultra-discretization: 
$\times\rightarrow+,+\rightarrow\min$, this formula agrees with the one
in Example \ref{ex:K}(1). Setting $x_i=y_i,s=t$, the geometric $K$-matrix for this
case is obtained as $\mathcal{K}(X,s)=(\tilde{X}^\vee,s)$ where $\tilde{X}=(x_{i+1})_i$ 
if $X=(x_i)_i$.
\item Let $n=5,k=2$. In this case $\mathcal{K}(X,s)=(\tilde{X}^\vee,s)$ is given by
the following formulas.
\begin{align*}
\tilde{x}_{11}&=x_{12}Q_1, \quad\tilde{x}_{12}=\frac{x_{13}x_{14}x_{23}}{Q_1Q_2}, 
\quad\tilde{x}_{13}=\frac{x_{11}Q_2}{x_{22}x_{23}}, \\
\tilde{x}_{22}&=Q_1, \quad\tilde{x}_{23}=\frac{x_{14}Q_3}{Q_1Q_2},
\quad\tilde{x}_{24}=\frac{x_{25}}{x_{14}}, \\
Q_1&=x_{13}+x_{24},\quad Q_2=x_{12}+x_{23},\quad Q_3=x_{12}x_{13}+x_{12}x_{24}+x_{23}x_{24}.
\end{align*}
Examples of the combinatorial $K$-matrix for $B^{2,1}$: 
$\small{\tab{1\\5}}\mapsto\small{\tab{3\\4}^\vee},
\small{\tab{2\\3}}\mapsto\small{\tab{1\\5}^\vee}$ 
in Example \ref{ex:reflection} are checked by ultra-discretization of the above formulas.
\end{enumerate}
\end{example}

\section{Discussions}

In \cite{KOY} the case when the KR crystal is $B^{1,l}$ is treated and the 
corresponding solution is denoted by {\tt Rotateleft}. Imitating the construction
there, one can generalize box-ball systems with a boundary using the solution
of the reflection equation obtained in this note. It should also be noted that
the $K$-matrix in this case is the $q\to0$ limit of the one constructed as
an intertwiner of certain coideal subalgebra of $U_q(A_{n-1}^{(1)})$ \cite[\S6.3]{KOY2}.

One might wonder if the trick explained in this note could be applied to 
the case when the $R$-matrix is a linear operator on a certain vector space,
such as an intertwiner of the quantized enveloping algebra 
$U_q(\mathfrak{g})$. 
However, as far as the authors try, a naive extension does not work.

KR crystals have been constructed for all non-exceptional types of quantum 
affine algebras \cite{FOS}. Hence, it is natural to ask whether the technique 
in this note is generalized to other types than $A^{(1)}_{n-1}$. Unfortunately, in those
cases, the dual crystal $B^\vee$ coincides $B$ itself, except 
$B^{n,l}$ and $B^{n-1,l}$ for $D^{(1)}_n$ and $n$ is odd. If $B^\vee=B$, then from Corollary
\ref{cor:dual} (2), we have $R_{B^\vee,B}(b^\vee\ot b)=R_{B,B}(b^\vee\ot b)
=b^\vee\ot b$, since $R_{B,B}$ is the identity. Therefore, $K_B$ also turns
out the identity. So, if either $B_1$ or $B_2$ is self-dual, then the reflection equation 
turns out trivial or reduces to the 
set-theoretical inversion relation $R_{B_2,B_1}R_{B_1,B_2}=\mathrm{id}$.

\section*{Acknowledgments}
The authors thank Akihito Yoneyama for collaboration in their previous works,
Vincent Caudrelier, Rei Inoue, Robert Weston and Yasuhiko Yamada for 
kind interest and comments.
They are supported by Grants-in-Aid for Scientific Research No.~16H03922 from JSPS. 
A.K. is supported by Grants-in-Aid for Scientific Research 
No.~18K03452, 18H01141 and M.O. by No.~19K03426.

\end{document}